\documentclass[letter,11pt]{article}
\usepackage[utf8]{inputenc}
\usepackage{amsmath}
\usepackage{amsthm}
\usepackage{xcolor}

\usepackage[margin=1in]{geometry}

\usepackage{algorithm2e}
\usepackage{quotes}
\usepackage{thmtools}
\usepackage{thm-restate}

\usepackage{setspace}
\usepackage{amsmath,amsthm,amssymb,amsfonts,bbm}
\usepackage{url}
\usepackage{enumerate}
\usepackage{hyperref}
\usepackage{cleveref}

\let\originalleft\left
\let\originalright\right
\renewcommand{\left}{\mathopen{}\mathclose\bgroup\originalleft}
\renewcommand{\right}{\aftergroup\egroup\originalright}

\newcommand{\set}[1]{\{#1\}}
\newcommand{\st}{{\ :\ }}
\newcommand{\floor}[1]{\lfloor{#1}\rfloor}

\newcommand{\N}{\mathbb{N}}

\newcommand{\card}[1]{\left|{#1}\right|}

\newcommand\Prob[2]{{\Pr_{#1}\left[ {#2} \right]}}

\newcommand{\yes}{{\textnormal{\sf{YES}}}}
\newcommand{\no}[1]{{\textnormal{\sf{NO}}}_{#1}}
\newcommand{\dham}{{d_{{\textnormal{\sf{HAM}}}}}}

\newtheorem{theorem}{Theorem}[section]
\newtheorem{lemma}[theorem]{Lemma}
\newtheorem{proposition}[theorem]{Proposition}

\newtheorem*{theorem*}{Theorem}
\newtheorem*{claim}{Claim}
\newtheorem*{lemma*}{Lemma}

\theoremstyle{remark}

\newtheorem*{observation}{Observation}

\theoremstyle{definition}
\newtheorem{definition}{Definition}
\newtheorem*{definition*}{Definition}

\title{Stochastic Distance in Property Testing}
\author{
{Uri Meir \thanks{Tel-Aviv University, Israel. Email: \texttt{urimeir@mail.tau.ac.il}}}
\and
{Gregory Schwartzman \thanks{JAIST, Japan. Email: \texttt{greg@jaist.ac.jp}}}
\and
{Yuichi Yoshida \thanks{NII, Japan. Email: \texttt{yyoshida@nii.ac.jp}}}
}

\date{}

\begin{document}

\maketitle

\thispagestyle{empty} 
\setcounter{page}{0}
\begin{abstract}
We introduce a novel concept termed "stochastic distance" for property testing. Diverging from the traditional definition of distance, where a distance $t$ implies that there exist $t$ edges that can be added to ensure a graph possesses a certain property (such as $k$-edge-connectivity), our new notion implies that there is a \emph{high probability} that adding $t$ \emph{random} edges will endow the graph with the desired property. While formulating testers based on this new distance proves challenging in a sequential environment, it is much easier in a distributed setting. Taking $k$-edge-connectivity as a case study, we design ultra-fast testing algorithms in the CONGEST model. Our introduction of stochastic distance offers a more natural fit for the distributed setting, providing a promising avenue for future research in emerging models of computation.
\end{abstract}

\newpage
\section{Introduction}
Property testing has become a major focus in computational research over the years. One of the main goals in this field is to quickly determine if a given structure has a specific property or is far from having it. Traditionally, this "distance" from a property has been defined using the Hamming distance, which counts the number of changes needed to give the structure the desired property. But a question arises: Is this the best way to measure distance in all situations, especially in emerging models of computation?

Consider the following motivating example. Distributed dynamic systems, such as peer-to-peer networks, frequently experience changes in their structure as nodes (clients) join or depart and edges may experience failures. For the sake of resilience against failures, it is imperative for these systems to maintain a topology with certain advantageous features, like $k$-edge-connectivity\footnote{Going forward we simply write $k$-connectivity.}. However, maintaining this throughout the evolution of the network may require resource-intensive corrections, which ideally should be minimized. Still, it might be the case that some topologies are "easy" to correct, in the sense that simply adding a small number of random edges to the graph will make the network topology $k$-connected. This can be seen as a "low-cost" fixing operation, compared to executing an elaborate algorithm that guarantees $k$-connectivity for \emph{every} topology. A question arises: can we detect these easy-to-fix topologies \emph{fast}?

We take a first step in addressing the above and introduce a new distance measure which we call \emph{stochastic distance}. That is, we say that a graph $G = (V,E)$ is $t$-stochastically-close to a property $\mathcal{P}$ if it holds with high probability (w.h.p)\footnote{With probability at least $1-n^{-c}$ for some constant $c>1$. The choice of $c$ does not affect the asymptotics of our results.} that adding (roughly) $t$ random edges to $G$ will make it have property $\mathcal{P}$. Intuitively, the Hamming distance metric can be seen as asking whether the input graph is close to \emph{at least one} graph instance that has a desired property, while our new distance measure asks whether the graph is close to \emph{many instances} that have the property.

To exemplify the usefulness of our new distance measure, we take the property of $k$-connectivity as a case study. This is an extremely desirable graph property in the distributed setting, as it guarantees that the system remains connected even under several edge failures. In Section~\ref{sec:connectivity} we note that the simple case of connectivity already becomes hard to test in the sequential setting. However, using the distributed power of the system we can design extremely fast distributed testers for both connectivity and $k$-connectivity. 

\subsection{Our model and results}
In the distributed setting, a network of nodes, which is represented by a communication
graph $G = (V, E)$, aims to solve some graph problem with respect to $G$. Every node in $G$ has unique ID of $O(\log n)$ bits.
Computation proceeds
in synchronous rounds, in each of which every vertex can send a message to each of its
neighbors. The running time of the algorithm is measured as the number of communication
rounds it takes to finish. Our results hold for the CONGEST model of distributed computation, where messages are limited to $O(\log n)$ bits (where $n = |V|$).

A distributed 1-sided tester \cite{Censor-HillelFS19} for a property $\mathcal{P}$ (or simply a tester) has the following guarantee. If the graph $G$ has the property, then all nodes accept. If the graph is $\epsilon$-far from having the property, then with probability larger than 2/3 \emph{at least} one node rejects. This definition remains the same for both the Hamming distance metric and our stochastic distance measure. 

Note that $\epsilon \in (0,1)$ indicates the distance from the property relative to the number of edges in the graph (i.e., $\epsilon |E|$ in the general model) or to the possible number of edges in the graph (i.e., $\epsilon \binom{n}{2}$ in the dense model). In our model every non-edge is added with probability $t / (\binom{n}{2} - |E|)$ (see Section~\ref{sec: prelims} for an exact definition), therefore it is more natural to use $t\in \mathbb{N}$ as the distance parameter rather than $\epsilon$. That is, our testers decide whether the graph has a certain property or if it is $t$-far from it. We state and prove the following two theorems\footnote{Where $\Tilde{O}$ subsumes factors logarithmic in $n$.}:

\begin{restatable}{theorem}{connectivityTester}
\label{thm:connectivity_tester}
There exists a deterministic algorithm in the CONGEST model, that for a parameter $s\in \mathbb{N}$ runs in $O(s)$ rounds and distinguishes whether the graph $G$ is connected, or is $\Omega((n \log n)/s)$-stochastically-far.
\end{restatable}

\begin{restatable}{theorem}{kconnectivityTester}
There exists a randomized algorithm in the CONGEST model, that for a parameter $s\in \mathbb{N}$ runs in $\Tilde{O}(s^4)$ rounds and distinguishes w.h.p whether the graph $G$ is $k$-connected, or is $\Omega((kn \log n)/s)$-stochastically-far from being one. 
\end{restatable}

Intuitively, the above states that the more random edges are added the easier it is to check if the graph will become $k$-connected or not. That is, if $s\approx kn\log n$ then checking whether a constant number of random edges will make the graph $k$-connected is as hard as checking if the graph is connected, which requires traversing the entire graph. If $s=O(1)$, then so many edges are added that the graph almost surely becomes $k$-connected. Our results provide a smooth transition between these two cases. 

\paragraph{Related work}
In (non-distributed) property testing, the objective is to devise algorithms that can distinguish between graphs satisfying a property $\mathcal{P}$ and graphs that are $\epsilon$-far from having the property with a high probability. Testing connectivity properties in the bounded-degree model, where we have query access to the input graph via a list of incidence lists, has been extensively studied, including $k$-connectivity~\cite{GoldreichR02,yoshida2010testing,orenstein2010testing}, $k$-vertex-connectivity~\cite{yoshida2012property,orenstein2010testing}, $(k,l)$-sparsity~\cite{ito2012constant} and supermodular-cut conditions~\cite{tanigawa2015testing}.

The first distributed property testing algorithm was due to \cite{BrakerskiP11}. A thorough study of distributed property testing was initiated in \cite{Censor-HillelFS19}. This was followed by a long line of work, presenting new and improved testers for various properties \cite{AugustineMPV22, EvenFFGLMMOORT17, FraigniaudRST16,FraigniaudHN20,LeviMR21,FichtenbergerV18}. All of these works consider the Hamming distance handed down from the sequential testing model.

$k$-connectivity received a large amount of attention in the distributed literature \cite{Thurimella97, Censor-HillelD20, Dory023,Dory18, DagaHNS19,Parter19, BezdrighinE0GHI22}. There is a large body of work that aims to find \emph{sparse connectivity certificates} -- for a $k$-connected graph the goal is to find a $k$-connected subgraph that has $O(kn)$ edges. Another related problem is computing a $k$-edge-connected spanning subgraph ($k$-ECSS) -- for a $k$-connected graph the goal is to find the \emph{sparsest possible} $k$-connected subgraph. 
We note that both these problems assume that the input graph is $k$-connected, and are therefore inherently different than the problem we consider.

\section{Preliminaries} 
 \label{sec: prelims}

\paragraph{Distance from a property}
We identify simple graphs $G=(V = [n],E)$ with the $\binom{n}{2}$-dimensional characteristic vector of their edge set $E\subseteq \binom{V}{2}$, and use the Hamming metric over these vectors, defined by $\dham(G, G') = \card{\set{i \in \binom{V}{2} \st G(i) \neq G'(i)}}$, where $G(i)$ is the $i^{th}$ entry in the vector $G$.

By extension, the distance of a graph $G$ from a family of graphs $\mathcal{G}$ is:
\[
    \dham(G,\mathcal{G}) = \min_{G'\in\mathcal{G}} \dham(G,G').
\]

A property $\mathcal{P}$ is formally a family of graphs (e.g., all connected graphs). 
Given positive integers $n$ and $t$, we define the $\yes$ case of the corresponding testing problem to be all graphs $G=(V,E)$ in $\mathcal{P}$ over $n$ vertices:
\[
\yes := \set{G \in \set{0,1}^{\binom{V}{2}} \st G\in\mathcal{P}}
\]
We define the $\no{}$ case, denoted $\no{t}'$, to be all graphs $G=(V,E)$ over $n$ vertices with Hamming distance at least $t$ from $\yes$:
\[
\no{t}' := \set{G \in \set{0,1}^{\binom{V}{2}}   \st \dham(G,\yes) \geq t} 
\]

Motivated by connectivity, in the following we only consider monotone non-decreasing properties, 
for which it is easy to see that only additions count 
towards the distance from the property (note that many properties in the literature are monotone non-increasing instead. For these, one can replace additions with deletions). 
Formally, if $\mathcal{P}$ is a monotone non-decreasing property, and $G \notin \mathcal{P}$ is a graph that does not satisfy it, then for any $H\in\mathcal{P}$ with $\dham(G,H) = \dham(G,\mathcal{P})$, we have $E_G \subseteq E_H$ (or alternatively, $\dham(G,H) = \card{\set{i \in \binom{V}{2} \st G(i) = 0 \wedge H(i) = 1}}$).

\paragraph{Stochastic distance} We define \emph{stochastic closeness} to a property as follows:
\begin{definition}[Random addition of edges]
        For a graph $G = (V,E)$, we choose a random subset $E'$ of the edge set $\bar{E} = \binom{V}{2} \setminus E$, by adding each edge with probability $t/\card{\bar{E}}$ independently, for parameter $t\in\left[0,\card{\bar{E}}\right]$. We define the random graph $\mathrm{Add}(G,t) := (V, E \cup E')$. We say that $G' = \mathrm{Add}(G,t)$ is created from $G$ by a \emph{random addition of edges} with parameter $t$.
    \end{definition}
The parameter $t$ should be understood intuitively as the (expected) number of random edges required to make the graph have the property $\mathcal{P}$ w.h.p.
Per our motivation, and in accordance with previous results for testing connectivity, we allow ourselves the mild assumption that $\card{\bar{E}} = \Omega(n^2)$ which clearly holds for any input relevant to our setting. The amount of edges we aim to add, however, is always significantly smaller, $t = o(n^2)$. In particular it is always the case that $t \leq \card{\bar{E}}$, as desired.

   
\begin{definition}[Stochastic closeness to a monotone property]
    \label{def:stochastic_dist}
    For a monotone property $\mathcal{P}$, a graph $G\notin \mathcal{P}$ over $n$ vertices is said to be $t$-stochastically-close to satisfying $\mathcal{P}$ if the graph $G' = \mathrm{Add}(G,t)$ satisfies
    \[
         \Pr[G' \notin \mathcal{P}] \leq n^{-c}
    \]
    for some global constant $c > 1$. As shown in \Cref{sec:definition_robustness}, the constant $c$ can be chosen arbitrarily without affecting stochastic closeness by more than a constant factor.
\end{definition}

We are now able to define the alternative set of NO instances:
\[
    \no{t} := \set{G\in \set{0,1}^{\binom{V}{2}} \st G\notin\yes \text{ and $G$ is not $t$-stochastically-close to $\mathcal{P}$}}
\]
In this text our main focus is solving promise problems of type $(\yes,\no{t})$, rather than $(\yes, \no{t}')$.

\paragraph{Comparing the two notions of distance.}
As an illustrative example, consider two $n$-node graphs, both with exactly two connected components: in $G_1$ a constant-size component is disconnected from the rest of the graph, while in $G_2$ there are two components of the same size, $n/2$.
While both graphs are exactly one edge away from being connected (i.e, both have hamming distance $1$). The situation is quite different if edges are added randomly, instead of being handpicked, leading to different \emph{stochastic} distance. While a small number of edges ($\sim \log n$) are already likely to connect $G_2$, the same amount has only probability $o(1)$ to connect $G_1$.


\section{Warm-up: Connectivity with Stochastic Distance}\label{sec:connectivity}




In this section we deal with the graph connectivity property for a network that is not necessarily connected. This section aims to present the ideas used in the following section about $k$-connectivity. 
We use the following terminology: when a graph is disconnected, write $G = \bigcup_{i=1}^{m} C_i$ as the unique decomposition of $G$ into its connected components $C_1,\dots,C_m$, with $C_i = (V_i,E_i)$, and $s_i := \card{V_i}$ for their sizes.

We make the following observation:
\begin{observation}
    Let $G = \bigcup_{i=1}^{m} C_i$ be a disconnected graph (i.e., $m \geq 2$), and let $s = \frac{1}{m} \sum_{i\in [m]} s_i$, then $G$ is $O(\frac{n}{s})$-close to being connected in Hamming distance.
\end{observation}
Indeed, it holds that $s = \frac{1}{m} \sum_{i\in [m]} s_i = n / m$, which means that $m = n/s$.
As a single edge can be used to connect two components, adding $m-1 = n/s-1$ edges suffices to connect the graph.

The observation is also tight: there exist graphs that require exactly this number of additions.
Joined with a Markovian argument, one can deduce that for any graph that is far from connectivity, there exists \emph{many} small connected components. This is a key argument in the analysis of existing connectivity testers (using Hamming distance) \cite{GoldreichR02}.

The main part of this section deals with proving a statement of similar taste for stochastic distance. While the Hamming distance of a graph from being connected is dictated by the \emph{average} size of a connected component, for stochastic distance this is dictated by the \emph{minimum} size of a connected component. 
As a result, a graph that is far from being connected in stochastic distance is only guaranteed to have \emph{one} small connected component. Formally, we prove the following lemma:
\begin{lemma}
\label{lem:stochastic_dist_connectivity}
    Let $G = \bigcup_{i=1}^{m} C_i$ be a disconnected graph (that is, $m \geq 2$), with components of sizes $s_i = \card{V_i}$.
    Then $G$ is $O\left((n \log n) / s\right)$-stochastically-close to being connected, where $s = \min_{i\in [m]}\set{s_i}$.
\end{lemma}

\begin{proof}
    Recall that when considering stochastic distance every non-edge is added with some probability $p$. Consider the set of bad events $\set{B_k}_{k=1}^{\floor{m/2}}$, where $B_k$ is the event that there exists a set of exactly $k$ connected components $C_{i_1},\dots C_{i_k}$ such that none of the edges between $\bigcup_{j=1}^{k} C_{i_j}$ and the rest of the graph are added. 
    The important observation is that the graph stays disconnected if and only if one of these bad events occurs. We go on to bound the probability of each of these bad events.
    
    Fix $k$. Our goal is to bound $\Pr[B_k]$.
    There are $\binom{m}{k}$ ways to choose a subset of $k$ components. Fix one such choice $i_1,\dots i_k$ and denote the set of vertices in these components by $U = \bigcup_{j=1}^{k} V_{i_j}$.
    As each component is of size at least $s$, we have $\card{U} \geq ks$, applying the same reasoning to the remaining graph, made of the rest of the components: $\card{V\setminus U} \geq (m-k)s \geq ks$ nodes (the second inequality holds since $k \leq m/2$).
    Thus, we have $ks \leq \card{U} \leq n-ks$, and consequently we get:
    \[
        \card{U} \cdot \card{V\setminus U} = \card{U} \cdot (n-\card{U}) \geq ks (n-ks).
    \]
    where the inequality is true since $f(x) = x(n-x)$ is unimodal and symmetric on the interval $[ks,n-ks]$.
    The probability of $U$ staying disconnected from $V \setminus U$ is thus bounded by:
    \[
        (1-p)^{\card{U} (n-\card{U})} \leq e^{- p \card{U} (n-\card{U})} \leq e^{- p ks (n-ks)}
    \]
    And by a union bound over all choices of $U$ with $k$ connected components, we get:
    \[
        \Pr[B_k] \leq \binom{m}{k} e^{- p ks (n-ks)} \leq e^{-p ks (n-ks) + k \log m} \leq e^{-p ksn + p(ks)^2 + k \log n}
    \]
    Thus, by taking $p = 2(c+2)\log n /(sn)$, the expression in the exponent is bounded by:
    \[
        -\frac{2(c+2) \log n \cdot k s n}{s n} + \frac{2(c+2) \log n (ks)^2}{s n} + k \log n =
        (2(c+2)(ks/n -1) + 1) k\log n \leq -(c+1)k\log n
    \]
    where the last inequality uses $ks/n \leq (m/2)s/n \leq 1/2$. This, in turn, bounds the probability of the bad event by: 
    \[
        \Pr[B_k] \leq e^{-(c+1)k \log n} \leq \left(\frac{1}{n}\right)^{-(c+1)k}
    \]
    Union bounding over all values of $k$, we get
    \[
        \Pr\left[\bigvee_{k=1}^{\floor{m/2}}B_k\right] \leq \sum_{k=1}^{\floor{m/2}} n^{-(c+1)k} \leq \sum_{k=1}^{\infty} n^{-(c+1)k} \leq 2n^{-(c+1)} \leq n^{-c}
    \]
    where the second to last inequality uses the sum of an infinite geometric series with $n^{-(c+1)} \leq 1/2$, and the last simply uses $n\geq 2$.

    Finally, we get that $G$ is connected w.h.p. This implies that $G$ is $t$-stochastically-close to being connected, with $t = O\left((n \log n) / s\right)$ (recall that $\card{\bar{E}} = \Omega(n^2)$).
    
    
\end{proof}

    \paragraph{Tightness of ~\Cref{lem:stochastic_dist_connectivity}.}
    We focus on the smallest connected component, $V_i$ of size $s_i = s$. There are at most $s(n-s) \leq sn$ potential edges to connect $V_i$ to the rest of the graph. If we take $p' = c \log n /(4 sn)$, then the probability none of them is added satisfies:
    \[
        (1 - p')^{s(n-s)}
        \geq (1 - p')^{sn}
        \geq e^{-2p' sn}
        = e^{-c \log n / 2}
        = n^{-c/2}
        > n^{-c} ,
    \]
    where the second inequality uses $1 - x \geq e^{-2x}$, which holds for $x\in[0,1/2]$. 
    

    Using the counter-positive of Lemma~\ref{lem:stochastic_dist_connectivity}, a graph $G$ that is $O((n\log n)/s)$-stochastically-far from being connected is guaranteed to have a connected component of size \emph{at most} $O\left(s\right)$. We state the following theorem.

\connectivityTester*
    \begin{proof}
    Assume that the parameter $s$ is known to all nodes in the network.
    Each node runs a distributed DFS algorithm. Every execution is associated with an ID, which is simply the ID of the root of the DFS tree. When multiple DFS executions visit the same node, all execution are terminated, except that with the maximum ID. When a DFS execution visits $s$ nodes, it terminates and checks whether there is an outgoing edge from any visited node to any unvisited node. If there is, then it accepts (i.e., the graph is connected), otherwise it rejects (i.e., the graph is far from being connected). Correctness follows from~\Cref{lem:stochastic_dist_connectivity}. It is clear that there is no congestion and that the algorithm terminates in $O(s)$ iterations.

    \end{proof}


\section{Stochastic Distance from $k$-Connectivity}
\label{sec:k-connectivity}

This section deals with $k$-connectivity, for any $k\geq 1$. 
We aim to characterize stochastic distance of a graph in terms of cuts, focusing on a specific parameter-of-interest denoted by $s_k(G)$: the smallest size of a vertex set that has a cut strictly smaller than $k$. In terms of this parameter, the following can be shown:

\begin{theorem}
\label{thm:stochastic_dist_k-connectivity}
    Let $G = (V,E)$ and $s = s_k(G)$, then $G$ is $O\left((k\cdot n \log n) / s\right)$-stochastically-close to being $k$-connected.
\end{theorem}

Similarly to \Cref{lem:stochastic_dist_connectivity}, this theorem is tight, up to a factor of $k$.
The rest of this section deals with proving the above theorem. First the parameter $s_k(G)$ is formally defined and discussed. Later, we generalize Lemma~\ref{lem:stochastic_dist_connectivity} to bound the number of random edge additions required to increase the connectivity of a graph by one. In the last part we formally prove Theorem~\ref{thm:stochastic_dist_k-connectivity}.

\subsection{Properties of minimal cuts}
\label{subsec:preparation}
For any vertex set, $U\subseteq V$, we write the cut size of $U$ as $c(U) := \card{E(U, V\setminus U)}$. We use $d(t) := t(n-t)$ for the \emph{potential} amount of edges between any $t$ vertices and the rest of the graph, and for a specific vertex set, $U$, we slightly abuse notation and write $d(U) = d(\card{U})$.

For any parameter $k\in\N$ and graph $G = (V,E)$, we define the collection of small cuts as the cuts that are strictly smaller than $k$:
\[
    \mathcal{S}_k(G) := \set{U \subseteq V \st c(U)<k \text{ and } U\neq\emptyset}
\]
A value of interest for us will be the size of the smallest set in such a collection. Formally we define:
    \[
        s_k(G):= \min_{U\in \mathcal{S}_k(G)} \card{U} 
    \]
When the connectivity parameter $k$ and the graph $G$ are clear from context, we omit either one or both (writing $s_k, s(G)$ or simply $s$ instead).
If the collection is empty for some $k$ and $G$, we define $s_k(G) = n$. Note that $s_k(G) = n$ if and only if $G$ is indeed $k$-connected.

We show that $s_k(G)$ has certain monotonicity properties with respect to the parameters $G$ and $k$. Indeed, by definition a larger value for $k$ creates a larger collection of cuts, i.e., $\mathcal{S}_{k}(G) \subseteq \mathcal{S}_{k+1}(G)$. On the other hand, adding edges to the graph can only increase the cut of any given vertex set, $U$. Thus, any $U$ with a small cut after additions, also had a small cut before additions. Formally, if $G \subseteq G'$ (i.e., $G=(V,E), G'=(V,E')$ and $E \subseteq E'$), we have $\mathcal{S}_k(G') \subseteq \mathcal{S}_k(G)$.
The minimum over elements in a collection can only decrease if we add elements to the collection, and hence we have:
\begin{observation}[Monotonicity of $s_k(G)$]
    \label{parameter_monotonicity}
    The parameter $s_k(G)$ is monotone non-increasing in $k$, and monotone non-decreasing in $G$.
\end{observation}

Below, we will be interested in finding the smallest vertex set $U$ in the collection $\mathcal{S}_k(G)$. This element determines the value $s$. 
The following proposition shows a property of such a set $U$ that will be useful for the algorithm:
\begin{proposition}
    Fix $k\in\N$ and a graph $G = (V,E)$. For any subset $U\subseteq V$ with a cut less than $k$ ($c(U) < k$) and of minimal size ($\card{U} = s_k$) the induced subgraph $G[U]$ is connected.
\end{proposition}

\begin{proof}
    Assume otherwise, then we can break $G[U]$ into its connected components: $U = \bigcup_{i} U_i$. Each such component has strictly less nodes than $U$. Moreover, we have:
    \begin{align*}
        c(U_i)
        = E(U_i, V\setminus U_i)
        = E(U_i, V\setminus U)
        \leq E(U, V\setminus U) 
        = c(U) ,
    \end{align*}
    where the second equality holds since there are no edges between different components of $U$.
    We conclude that each $U_i$ has a cut of size at most $c(U) < k$, but their size is strictly smaller than $s$, which leads to a contradiction.  
\end{proof}

\subsection{Increasing the connectivity by one}
\label{subsection:increase_connectivity}
Using the notations of this section, the statement of Lemma~\ref{lem:stochastic_dist_connectivity} is that any graph that is ``$0$-connected'' (i.e., disconnected) is $O\left(n \log(n) / s_1\right)$-stochastically-close to being $1$-connected. This is true since $s_1$ is exactly the size of the smallest connected component in $G$.

We generalize this statement for $k$-connectivity:

\begin{lemma}
    \label{lem:increase_connectivity_whp}
    Any $(k-1)$-connected graph $G$ with $n \geq 4k$ nodes is $O\left(n \log(n) / s_k\right)$-stochastically-close to being $k$-connected.
\end{lemma}


\begin{proof}
    For $k = 1$, the proof is complete due to Lemma~\ref{lem:stochastic_dist_connectivity}. 
 For $k > 1$, let $G$ be a $(k-1)$-connected graph. Thus, all minimum cuts have $k-1 \geq 1$ edges ($G$ is connected). We estimate the number of random edges needed to make all sets in $\mathcal{S}_k(G)$ have a cut of size at least $k$.
 
    
    We identify minimum cut with a partition of $V$ to $U$ and $V\setminus U$ (w.l.o.g, $\card{U} \leq \floor{n/2}$). As in Lemma~\ref{lem:stochastic_dist_connectivity} assume that every non-edge is added with probability $p$.
    Consider the event $B_U$ that $U$ remains with cut of size $k-1$. That is, not a single edge was added between $U$ and $V\setminus U$. There are exactly $d(U) - (k-1)$ potential additions, and so:
    \[
        \Prob{}{B_U} = (1-p)^{d(U)-(k-1)}
        = (1-p)^{\card{U}(n-\card{U}) -(k-1)}
        \leq (1-p)^{s_k\cdot n/4}
        \leq e^{-s_k\cdot pn/4} ,
    \]
    where the first inequality uses $n-\card{U} \geq n/2\ , \ \card{U} \geq s_k(G)$ and $k \leq \card{U}n/4$ to lower bound the exponent.
    
    By taking $p = 4(c+2) \log n /(s_k n)$, this bound becomes
    \[
        e^{-s_k\cdot pn/4}
        = e^{-(c+2)\log n}
        = n^{-(c+2)}
    \]
    
    We finish the proof by using a union bound over all minimum cuts. There are at most $\binom{n}{2} \leq n^2$ of them due to well known corollary of Karger's Algorithm \cite{Karger93}. The probability that $G$ is not $k$-connected after the edge additions is at most
    \[
        \sum_{U\in\mathcal{S}_k(G)} \Prob{}{B_U}
        \leq n^2 \cdot n^{-(c+2)} = n^{-c} .
        \qedhere
    \]
\end{proof}

\subsection{Putting it all together}
\label{subsection:conclude_k_connectivity}

We are now ready to prove Theorem~\ref{thm:stochastic_dist_k-connectivity}.
Our proof considers an alternative addition process which is easier to analyze. Intuitively, let $p' \approx p/k$ and consider the following process. ``Repeat $k$ times: add each edge to the graph independently with probability $p'$''. We formalize the above and bound the required value of $p'$ for the alternative process, analyzing each of the $k$ iterations separately with Lemma~\ref{lem:increase_connectivity_whp}. We finish up by relating the original addition processes to the alternative one.



\begin{proof}[Proof of Theorem~\ref{thm:stochastic_dist_k-connectivity}]
    Denote by $r<k$ the connectivity of $G$ (if $G$ is disconnected, then $r=0$).
    We consider $k-r$ (at most $k$) iterations of additions, which we enumerate by $r+1,\dots, k$. We further denote by $G_i$ the graph after iteration $i$ (and $G_r = G$).
    
    First we consider process (A), where at each iteration $i\in\set{r+1,\dots,k}$, each edge is independently added with probability 
    $$ p_i = 4(c+3)n\log n / s_{i+1}(G_{i}) .$$
    Applying Lemma~\ref{lem:increase_connectivity_whp} for each iteration $i$ we have the following:
    The probability that $G_i$ is $i$-connected, \emph{given that $G_{i-1}$ is $(i-1)$-connected} is at least $1 - n^{-(c+1)}$.
    Union bounding over all $k-r \leq n$ iterations, we have that with probability at least $1- n^{-c}$, the final graph $G_k$ is $k$-connected.
    
    Next, consider an adjusted process (B) that also works in iterations, but uses a fixed value
    $$p' = 4(c+3)n\log n / s_{k}(G) .$$
    
    Using the monotonicity of $s_k(G)$, we have
    \[
        s_k(G) = s_k(G_r) \leq s_{i+1}(G_r) \leq s_{i+1}(G_{i}) ,
    \] 
    and thus $p_i \leq p'$ for all values of $i$.
    
    Let us focus on a single non-edge of $G$, call it $(u,v)$, and follow it through the entire process.
    The probability $(u,v)$ is added to $G_k$ is higher in process (B) than it is in process (A), since $p' \geq p_i$ for all $i$. On the other hand, this probability is at most $p'k$, by a union bound over all iterations of process (B).
    
    Lastly, consider process (C), a one-shot random addition, using
    $$p = p'k = 4(c+3)k n \log n / s_{k}(G) .$$
    The probability of every non-edge being added in process (C) is higher than it is in process (B) which is in turn higher than process (A). 
    By monotonicity of $k$-connectivity (as a graph property), the probability that $G_k$ is a $k$-connected graph is therefore also higher in process (C) than it is in process (A), and is therefore at least $1 - n^{-c}$.
\end{proof}

The counter-positive of Theorem~\ref{thm:stochastic_dist_k-connectivity} means that any graph $G$ which is $O\left((k\cdot n \log n) / s\right)$-stochastically-far from being $k$-connected has a set $W$ of at most $s$ nodes with a cut smaller than $k$. We call such set $W$ an $(s,k)$-\emph{witness}, or simply a \emph{witness} when $k$ and $s$ are clear from context.

\section{Distributed algorithm for $k$-connectivity}
In this section, we provide a distributed algorithm that detects an $(s,k)$-witness within $O(s^4 \log n)$ rounds in the CONGEST model w.h.p. 
We state with the following useful lemma\footnote{The original lemma in \cite{GoldreichR02} uses the notion of a \emph{$j$-extreme} node set which is equivalent to our notion of a $(s,k$)-witness.}:
\begin{lemma}[Lemma 3.16 in~\cite{GoldreichR02}]
    \label{lem:cheap_spanning_tree}
    Let $W$ be an $(s,k)$-witness. Suppose that each edge in the graph is independently assigned a uniformly distributed cost in $[0,1]$. Then, with probability at least $\Theta\left(s^{-2(1-1/k)}\right)$, $W$ contains a spanning tree such that every edge in the tree has cost smaller than any edge in the cut $E(W, V\setminus W)$.
\end{lemma}

Specifically, if the tree guaranteed by the above lemma exists, the MST must also have the same guarantees. Furthermore, as all edge costs are unique w.h.p., the MST is unique w.h.p. When referring to the above lemma, it will be convenient to refer to a single tree, the MST induced by the cost function (i.e., when referring to the spanning tree guaranteed by Lemma~\ref{lem:cheap_spanning_tree}, we always refer to an MST).

Given a starting vertex $w$ and a size bound $s$, consider the following random process: 
\begin{itemize}
    \item Input: start vertex $u\in V$, size bound $s\in \N$.
    
    \item Assign uniformly random weights to all edges: $w: E \to [0,1]$.
    
    \item Start with a singleton set $W = \set{u}$. As long as $\card{W} < s$, repeat the following:
    \begin{itemize}
        \item Choose the smallest cut edge $e = \mathrm{argmin}_{e \in E(W, V \setminus W)} w(e)$.
        
        \item Update $W$ to include the new node of $e$ in $V \setminus W$.
        
        \item If $c(W) < k$ then declare $W$ as an $(s,k)$-witness.

    \end{itemize}
     
\end{itemize}

Lemma~\ref{lem:cheap_spanning_tree} implies that repeating the above procedure $\Tilde{\Theta}(s^2)$ times detects a witness w.h.p.

In~\cite{GoldreichR02} a small amount of random nodes is sampled, and the above procedure is applied for each of them. This is possible since any graph that is far from $k$-connectivity in \emph{Hamming} distance contains many witnesses.
In our case, a graph that is \emph{stochastically} far only guarantees a single witness, and therefore it is costly to detect the witness in the standard (query) model. We leverage the distributed nature of the system to apply the witness detection procedure to all nodes in parallel.

\paragraph{Distributed implementation}

The above algorithm admits a simple distributed implementation. First, we assign random weights in $[0,1]$ to all edges in the graph, this guarantees that with some small probability there exists a spanning tree within the witness component satisfying the guarantees of Lemma~\ref{lem:cheap_spanning_tree}. Assume this tree exists and denote this tree by $T$. It will be clear from our algorithm that if this tree does not exist, the running time does not change however we may simply not detect the witness component if it exists.

Every node grows a cluster as follows. Every cluster can be implemented via a tree rooted at a leader node. All nodes in the cluster know for each neighbor whether it is in the cluster or not. To pick the next edge to be added to the cluster the edge weights are propagated along the tree towards the root, taking the minimum at every node. Finally, the root propagates the minimum value to all nodes in the cluster and the corresponding edge is added to the cluster. After each edge addition the cluster checks whether the current cut size is smaller than $k$ and terminates if this is the case (a witness is found).
 For ease of analysis let us define an \emph{iteration} as the step of adding a single edge (and vertex) to the cluster. When running this procedure from multiple nodes we can have every cluster wait for $\Theta(i)$ rounds to complete the $i$-th iteration. This will guarantee when running this procedure from multiple nodes that the iterations are executed in lock-step. We execute this process for $s$ iterations. This results in a running time of $O(s^2)$ rounds.

In order to detect a witness component, we must execute the above algorithm to completion, starting from a node within the witness component. We do not know which nodes are within the witness component before executing the algorithm, and running the procedure to completion may cause congestion. To overcome this, we execute the procedure from all nodes simultaneously, but prioritize executions that overlap according to a priority condition.

The priority of the cluster at iteration $i$ is the largest edge cost that was added to the cluster so far. Let us call this the \emph{max-edge} value.
Whenever multiple executions reach the same node, executions with higher max-edge values are terminated (ties are broken via the cluster root IDs). As clusters add edges greedily according to edge weights, we know that only edges from $T$ will be added to clusters whose execution starts from the witness component. That is, all edges in $T$ will be added, at which point the small cut will be detected and all nodes in the cluster will reject. 

As we terminate executions that visit the same node, we never experience congestion. It is sufficient to prove that at least one cluster execution that starts within the witness component survives. As all edges in the tree have weights smaller than the cut edges this means that the execution of clusters that start within the witness component is not affected by clusters that start from outside the component. This is because they always have a lower max-edge value due to guarantees on the weights of $T$.
Furthermore, in every iteration at least one cluster in the witness component is not terminated. That is, the cluster with highest priority inside the witness component at iteration $i$ must survive. We conclude that there exists a cluster construction that start within the component and terminates within $O(s^2)$ rounds. To get the guarantee of \Cref{lem:cheap_spanning_tree} w.h.p., we must repeat the process $\Tilde{O}(s^2)$ times. Thus, the final running time is $\Tilde{O}(s^4)$. We state the following theorem:

\kconnectivityTester*
\bibliographystyle{alpha}
\bibliography{refs}
\appendix
\section{Robustness of Stochastic Closeness}
\label{sec:definition_robustness}
The following claim is used to show robustness of stochastic closeness (\Cref{def:stochastic_dist}) with respect to the choice of the global constant $c$. Intuitively speaking, by repeating the random addition only a constant number of times, the probability of not attaining the property $\mathcal{P}$ can be reduced to any small polynomial.
\begin{claim}
    Fix a monotone property $\mathcal{P}$, a graph $G\notin \mathcal{P}$ over $n$ vertices, and constant $c > 1$.
    If for some $t\in\left[0,\card{\bar{E}}\right]$, it holds that
    \[
         \Pr[\mathrm{Add}(G,t) \notin \mathcal{P}] \leq n^{-c} ,
    \]
    then for any $m\in \N$ such that $m t \in\left[0,\card{\bar{E}}\right]$, it holds that
    \[
         \Pr[\mathrm{Add}(G,mt) \notin \mathcal{P}] \leq n^{-mc} .
    \]
\end{claim}

\begin{proof}
    Recall that in the randomly augmented graph $\mathrm{Add}(G,t)$ each edge $e\in \bar{E}$ is added independently with probability $t/|\bar{E}|$.
    Let $G_1,\dots, G_m$ be $m$ independent copies of $\mathrm{Add}(G,t)$. By definition, we have
    \[
        \Pr\left[\bigwedge_{i=1}^{m}\  G_i \notin \mathcal{P}\right] 
        = \prod_{i=1}^{m} \Pr\left[ G_i \notin \mathcal{P} \right]
        \leq n^{-mc} .
    \]
    Define $G' = \bigcup_{i=1}^{m} G_i$, the union of all $m$ randomly augmented graphs. $G'$ trivially contains all edges of $G$.
    On the one hand, the probability that each edge $e\in \bar{E}$ appears in $G'$ is at most $(mt)/|\bar{E}|$, by a union bound over the $m$ attempts to add it. Since the property $\mathcal{P}$ is monotone, we have
    \[
        \Pr\left[Add(G,mt) \notin \mathcal{P}\right]
        \leq \Pr\left[G' \notin \mathcal{P}\right] .
    \]
    On the other hand, $G'$ is a supergraph of $G_i$ for all $i\in[m]$. Using again the monotonicity of $\mathcal{P}$, we can write
    \[
        \Pr\left[G' \notin \mathcal{P}\right]
        = \Pr\left[\bigwedge_{i=1}^{m}\  G' \notin \mathcal{P}\right]
        \leq \Pr\left[\bigwedge_{i=1}^{m}\  G_i \notin \mathcal{P}\right]
        \leq n^{-mc} ,
    \]
    which concludes the proof.
\end{proof}
\end{document}